\documentclass[a4paper, numberwithinsect, cleveref, autoref, thm-restate]{lipics-v2021}

\title{Universal gauge-invariant cellular automata}

\author{Pablo Arrighi}{Université Paris-Saclay, CNRS, LMF, 91190 Gif-sur-Yvette, France \and IXXI, Lyon}{}{https://orcid.org/0000-0002-3535-1009}{} 

\author{Marin Costes}{ENS Paris-Saclay, CNRS, LMF, 91190 Gif-sur-Yvette, France}{marin.costes@ens-paris-saclay.fr}{https://orcid.org/0000-0003-0915-9192}{}

\author{Nathana\"el Eon}{Aix-Marseille Universit\'e, Universit\'e de Toulon, CNRS, LIS, Marseille, France}{}{https://orcid.org/0000-0003-0649-0891}{}

\authorrunning{P. Arrighi, M. Costes and N. Eon} 

\Copyright{Pablo Arrighi, Marin Costes and Nathana\"el Eon} 



\acknowledgements{The authors would like to thank Guillaume Theyssier for asking us the question whether any CA admits a gauge extension. This publication was made possible through the support of the ID\# 61466 grant from the John Templeton Foundation, as part of the “The Quantum Information Structure of Spacetime (QISS)” Project (qiss.fr). The opinions expressed in this publication are those of the author(s) and do not necessarily reflect the views of the John Templeton Foundation.}

\nolinenumbers 

\EventEditors{John Q. Open and Joan R. Access}
\EventNoEds{2}
\EventLongTitle{46th  International Symposium on Mathematical Foundations of Computer Science (MFCS 2021)}
\EventShortTitle{MFCS 2021}
\EventAcronym{MFCS}
\EventYear{2021}
\EventDate{August 23--27, 2021}
\EventLocation{Talinn, Estonia}
\EventLogo{}
\SeriesVolume{46}
\ArticleNo{5}

\usepackage{todonotes}
\usepackage{mathtools} 
\usepackage{amsmath}
 
\begin{document}

\maketitle

\begin{abstract} Gauge symmetries play a fundamental role in Physics, as they provide a mathematical justification for the fundamental forces. Usually, one starts from a non-interactive theory which governs `matter', and features a global symmetry. One then extends the theory so as make the global symmetry into a local one (a.k.a gauge-invariance). We formalise a discrete counterpart of this process, known as gauge extension, within the Computer Science framework of Cellular Automata (CA). We prove that the CA which admit a relative gauge extension are exactly the globally symmetric ones (a.k.a the colour-blind). We prove that any CA admits a non-relative gauge extension. Both constructions yield universal gauge-invariant CA, but the latter allows for a first example where the gauge extension mediates interactions within the initial CA.
\end{abstract}

\section{Introduction}
\label{sec:intro}

Symmetries are an essential concept, whether in Computer science or in Physics. In this paper we explore the Physics concept of gauge symmetry by taking it into the rigorous, Computer Science framework of Cellular Automata (CA). Implementing gauge-symmetries within CA may prove useful in the fields of numerical analysis; quantum simulation; and digital Physics---as these are constantly looking for discrete schemes that simulate known Physics. Quite often, these discrete schemes seek to retain the symmetries of the simulated Physics; whether in order to justify the discrete scheme as legitimate or as numerically accurate (e.g. by doing the Monte Carlo-counting right \cite{HastingsMonteCarlo}). More specifically, the introduction of gauge-symmetries within discrete-time lattice models has proven useful already in the field of Quantum Computation, where gauge-invariant Quantum Walks and Quantum Cellular Automata \cite{ArrighiQED} provide us with concrete digital quantum simulation algorithms for particle Physics. These come to complement the already existing continuous-time lattice models of particle Physics \cite{georgescu2014quantum,notarnicola2020real}. Another field where this has played a role is Quantum error correction \cite{kitaev2003fault,nayak2008non}, where it was noticed that gauge-invariance amounts to invariance under certain local errors. This echoes the fascinating albeit unresolved question of noise resistance within Cellular Automata \cite{harao1975fault,Toom}. 

In \cite{salo2013color} the authors study $G$--blind cellular automata, where $G$ is a group of permutation acting on the state space of cells. Blind cellular automata are globally symmetric under $G$, i.e. the global evolution commutes with the application of the same $g\in G$ at once on every cell. They show the surprising result that any CA can be simulated by such a globally symmetric CA, when $G$ is the symbol permutations. Globally symmetric CA are therefore universal. Local symmetry, aka gauge symmetry, is way more stringent however: a different $g_x$ can now be chosen for every cell $x$. Still, in this paper, we prove that any CA can be extended into a gauge-invariant CA. Gauge-invariant CA are therefore universal.

From a Physics perspective one usually motivates the demand for a certain gauge symmetry, from an already existing global symmetry. From a mathematical perspective, the gauge field that then gets introduced for that purpose is often seen as a connection between two gauge choices at neighbouring points. This raises questions however, because there is no immediate reason why a gauge symmetry should necessarily arise from an already existing global symmetry (one could ask for a certain ad hoc gauge symmetry from scratch). Nor is there an immediate reason why a gauge field should necessarily be interpretable as a connection (a gauge field could be made to hold absolute instead of relative information about gauge choices).\\
In this paper, we prove an original result relating these two folklore perspectives about gauge theories using purely combinatorial definitions. Namely, we prove that the CA that admit {\em relative} gauge extension are exactly those that have the corresponding global symmetry in the first place.

Although the gauge field was initially introduced in order to obtain gauge symmetry, it allows for new dynamics. Amongst those dynamics, one could ask for the matter field to influence the dynamics of the gauge field, as is the case in Physics. In this paper, we provide a first a Gauge-invariant CA where this is happening. This CA is obtained through a non-relative gauge extension. We leave it as an open question whether the same can be achieved though a relative gauge extension.

The present work builds upon two previous papers by a subset of the authors, which laid down the basic definitions of gauge-invariance for CA and provided a first set of examples, in both the abelian \cite{arrighi2018gauge} and the non-abelian \cite{arrighi2019non} cases.  Sec. \ref{sec:definitions} first recalls these basic definitions, but it also formalises the notions of general and relative gauge extensions, which were still missing. Sec. \ref{sec:extension1} shows that CA admit a relative gauge extension if and only if they are globally symmetric. Sec. \ref{sec:extension2} shows that any CA admits a general gauge extension. Sec. \ref{sec:universal} draws the consequences upon universality. Sec. \ref{sec:sourcing} provides a first example of a gauge-extended CA whose gauge field is sourced by the matter field.
\section{Definitions}
\label{sec:definitions}

\subsection{Cellular Automata}
A cellular automaton (CA) consist in an array of identical cells, each of which may take one in a finite number
of possible states. The whole array evolves in discrete time steps by iterating a function $F$. Moreover
this global evolution $F$ is shift-invariant (it acts everywhere the same) and causal (information cannot
be transmitted faster than some fixed number of cells per time step). Let us make this formal.   
\begin{definition}[Configuration]\label{def : Configuration}
A configuration c over an alphabet $\Sigma$ and a space $\mathbb{Z}^d$ is a function that associates a state to each point:
$$c\ :\ \mathbb{Z}^d\longrightarrow\Sigma$$
The set of all configurations will be denoted $\Sigma^{\mathbb{Z}^d} $.
\end{definition}

A configuration should be seen as the state of the CA at a given time. We use the short-hand notation $c_x=c(x)$ for $x\in \mathbb{Z}^d$ and $c_{|I}$ for the configuration c restricted to the set $I$---i.e.  $c:I\longrightarrow \Sigma$---for $I\subseteq\mathbb{Z}^d$. The association of a position and its state is called a cell.

The way to describe a global evolution $F$ that is causal, is via the provision of a local rule. A local rule takes as input a configuration restricted to $x+\mathcal{N}$, and outputs the next value of the cell $x$, i.e.  $f\ :\ \Sigma^\mathcal{N}\longrightarrow\Sigma$, where $\mathcal{N}$ is a finite subset of $\mathbb{Z}^d$ referred to as `the neighbourhood'. Applying $f$ at every position $x$ simultaneously, implements $F$.

\begin{definition}[Cellular automata]\label{def : Cellular Automata}
The CA $F$ having neighbourhood $\mathcal{N}$ and local function $f:\Sigma^\mathcal{N}\longrightarrow\Sigma$ is the function $F:\Sigma^{\mathbb{Z}^d}\longrightarrow\Sigma^{\mathbb{Z}^d}$ such that for all $x\in\mathbb{Z}^d$,
$$F(c)_x=f(c_{|x+\mathcal{N}}).$$
\end{definition}
We sometimes denote by $c_{t,x}$ the value of a cell at position $x$ and time $t$, where $c_{t+1}=F(c_t)$.

\subsection{Global versus gauge symmetry}

\paragraph*{Global symmetry}

We say that a CA is globally symmetric whenever its global evolution is invariant under the application of the same alphabet permutation at every position at once. Globally symmetric CA are also known as $G$-blind CA \cite{salo2013color} with $G$ a group of permutations over $\Sigma$. 
\begin{definition}[Globally symmetric]
    Let $F:\Sigma^{\mathbb{Z}^n} \rightarrow \Sigma^{\mathbb{Z}^n}$ be a CA and $G$ be a group of permutations over $\Sigma$. For all $g\in G$, let $\bar{\gamma}$ denote its application at every position simultaneously: $\bar{\gamma} (c)_i = g(c_i)$.
    We say that $F$ is globally $G$-symmetric if and only if, for any $g$ in $G$, we have 
    $F\circ \bar{\gamma} = \bar{\gamma}\circ F$.
\end{definition}

\paragraph*{Local/gauge symmetry}

We say that a CA is locally symmetric whenever its global evolution is invariant under the application of a local permutation at every position. The first difference with globally symmetric CA is the permutation is now allowed to differ from one position to the next. The second difference is that the permutation is now allowed to act on the surrounding cells. Locally symmetric CA are referred to as gauge-invariant CA \cite{arrighi2018gauge,arrighi2019non, ArrighiGaugeInvarianceRecap}. 

\begin{definition}[Local gauge-transformation group]\label{def : Local gauge-transformation group}
Let $g$ be a permutation over $\Sigma^{(2s+1)^d}$, with $s\in\mathbb{N}$. We denote by $g_x:\Sigma^{\mathbb{Z}^d}\longrightarrow\Sigma^{\mathbb{Z}^d}$ the function that acts as $g$ on the cells at $[\![x-s,x+s]\!]^d$, and trivially everywhere else. A {\em local gauge-transformation group} $G$ is a group of bijections over $\Sigma^{(2s+1)^d}$, such that for any $g,h\in G$ and any $x\neq y\in {\mathbb Z}^d$, $g_x\circ h_y=h_y\circ g_x$. 
\end{definition}
This permutation condition makes it irrelevant to consider which local gauge-transformation gets applied first, so that the product $g_x h_y$ be defined. The condition is decidable, checking it over the $[\![-2s,+2s]\!]^d$ suffices.

\begin{definition}[Gauge-transformation]\label{def : Gauge-transformations}
Consider $G$ a group of local gauge-transformations.
A {\em gauge-transformation} is then specified by a function $\gamma:\mathbb{Z}^d\longrightarrow G$. It is interpreted as acting over ${\cal C}$ as follows:
$$\gamma(c)=(\prod_{x\in\mathbb{Z}^d}\gamma_x)(c),$$
where $\gamma_x$ is short for $\gamma(x)$. We denote by $\Gamma$ the set of gauge-transformations.
\end{definition}
Notice how an element $\gamma\in\Gamma$ may be thought of as a configuration over the alphabet $G$---with $g_x$ the state at $x$.

Gauge-invariant CA are `insensitive' to gauge-transformations: performing $\gamma$ before $F$ amounts to performing some $\gamma'$ after $F$.
\begin{definition}[Gauge-invariant CA]\label{def : Gauge symmetry}
Let $F$ be a CA, $G$ be a local gauge-transformation group, and $\Gamma$ be the corresponding set of gauge-transformations. $F$ is $\Gamma$-gauge-invariant if and only if there exists a CA $Z$ over the alphabet $G$, such that for all $\gamma\in\Gamma$:
$$Z(\gamma)\circ F=F\circ\gamma.$$
\end{definition}
The reason why $\gamma'$ must result from a CA $Z$, instead of being left fully arbitrary, is because $F$ is deterministic, shift-invariant and causal---from which it follows that $\gamma'$, if it exists, can be computed deterministically, homogeneously and causally from the $\gamma$ applied before.
Thus, the above is demanding a weakened commutation relation between the evolution $F$ and the set of gauge-transformations $\Gamma$. In practice in Physics $Z$ is often the identity, making gauge-invariance a commutation relation. This will be the case in our constructions.

\subsection{(Relative) gauge extensions}

In Physics, one usually begins with a theory that explains how matter freely propagates, i.e. in the absence of forces. This initial theory solely concerns the `matter field', and is not gauge-invariant. For instance, the Dirac equation, which dictates how electrons propagate, is not $U(1)$--gauge-invariant. Next, one enriches the initial theory with a second field, the so-called gauge field, so as to make the resulting theory gauge-invariant. For instance, the case of the electron, $U(1)$--gauge-invariance is obtained thanks to the addition of the electromagnetic field. The resulting theory can still account for the free propagation of the matter field, but the presence of the gauge field also allows for richer behaviours, e.g. electromagnetism. Quite surprisingly three out of the four fundamental forces can be introduced mathematically, and thereby justified by gauge symmetry requirements, through this process of `gauge extension'. 

But when is it the case that a theory is a gauge extension of another, exactly? In Physics this is left informal. One of the contributions of this paper is the provide a first rigorous definition of the notion of gauge extension, and of its relative subcase, in the discrete context of CA.

\paragraph*{General gauge extension}

A gauge extension must simulate the initial CA, extend the required gauge-transformations, and achieve gauge-invariance overall:
\begin{definition}[Gauge extension]\label{def : Absolute-gauge extension}
Let $F$ be a CA over alphabet $\Sigma$. Let $\Gamma$ be a gauge-transformation group. Let $\Lambda$ be a finite set which will serve as the gauge field alphabet.
A gauge extension of $(F,\Gamma)$ is a tuple $(F', \Gamma')$ with $F'$ a CA over alphabet $\Sigma \times \Lambda$ and $\Gamma'$ a gauge-transformation group, such that:

\begin{itemize}
    \item \label{subdef:gaugeextsimulation} \emph{(Simulation)} there exists $\epsilon\in\Lambda$ such that $F'$ simulates one step of $F$ when the gauge field value is set $\epsilon$ everywhere. In other words for any $c\in \Sigma^{\mathbb{Z}^d}$, there exists $e'\in \Lambda^{\mathbb{Z}^d}$, $$F'(c,e) = (F(c),e')$$ 
    where $e$ is the constant gauge field configuration $(x\mapsto \epsilon)$. 
    \item \label{subdef:gaugeextextension}\emph{(Extension)} $\Gamma'$ extends $\Gamma$: there exists a bijection $B:\Gamma'\rightarrow \Gamma$ such that for any $\gamma'\in \Gamma'$, there exists a CA $L$ over alphabet $\Lambda$, such that for any $c, \lambda \in \Sigma^{\mathbb{Z}^d} \times \Lambda^{\mathbb{Z}^d}$,
    \begin{equation}\label{defeq : gaugetransfext}
        \gamma'(c,\lambda) = (\gamma(c),L(\lambda))
    \end{equation}
    where $\gamma = B(\gamma')$.
    \item \label{subdef:gaugeextinv}\emph{(Gauge-invariance)} $F'$ is $\Gamma'$-gauge-invariant.
\end{itemize}

\noindent We used implicitly the canonical bijection between $(\Sigma \times \Lambda)^{\mathbb{Z}^d}$ and $\Sigma^{\mathbb{Z}^d} \times \Lambda^{\mathbb{Z}^d}$
\end{definition}

Notice that when the gauge field does not evolve in time, we can rewrite the simulation condition as $F'(c,e) = (F(c),e)$. Then $F$ is a sub-automaton of $F'$ \cite{guillon:hal-00521624}, whenever the gauge field is set to $e$.

Intuitively, the gauge field's role is to keep track of which gauge-transformation got applied where, so as to hold enough information to insure gauge-invariance. There are different ways to do this; for instance one could indeed store the `gauge' at each point, i.e. which gauge-transformation has happened at the specific point. But one could be more parsimonious and store just the `relative gauge', i.e. which gauge-transformation relates that which has happened at every two neighbouring points. 

\paragraph*{Relative gauge extension}

The standard choice in the Physics literature is to place the gauge field between the matter cells only---i.e. on the links between two cells. This choice of layout is sometimes referred to as the `quantum link model' \cite{chandrasekharan1997quantum,silvi2014}. The mathematical justification for this choice, is precisely that the gauge field may be interpreted as relative information between neighbouring matter cells. Geometrically speaking, it may be understood as a `connection' relating two closeby `tangent spaces' on a manifold. 

Our previous definition of general gauge extensions does allow for such relative gauge extensions as a particular case, up to a slight recoding, as shown in Fig-\ref{fig:Equivalence between two different type of GE}, i.e. the link model is simulated by transferring the value of a gauge field on a link, to the vertex at the tip of the link.

\begin{figure}[ht!]
    \centering
    \begin{subfigure}[b]{0.49\textwidth}
    \centering
    \resizebox{\textwidth}{!}{\input{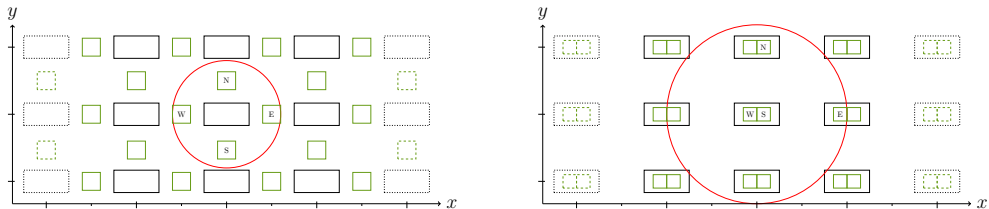}}
    \caption{The link model layout\ldots}
    \label{Link model GE}
    \end{subfigure}
    \begin{subfigure}[b]{0.49\textwidth}
    \centering
    \resizebox{\textwidth}{!}{\definecolor{gaugeCol}{rgb}{0.40390625,0.6109375,0.09109375}
\definecolor{evolutionCol}{rgb}{0.3,0.3,0.3}
\definecolor{gaugeCol}{rgb}{0.40390625,0.6109375,0.09109375}
\definecolor{evolutionCol}{rgb}{0.3,0.3,0.3}
\definecolor{grayCol}{rgb}{0.6,0.6,0.6}

\newcommand{\stateTikz}[6]{
  \draw[color=#5, #6, very thick](#1- #3 , #2) rectangle (#1 + #3, #2 + #4);
}

\newcommand{\gaugestate}[5]{
  \ifthenelse{#5 > 0}{\def\mycol{gaugeCol}}{\def\mycol{white}};
  \filldraw[fill=\mycol, very thick](#1- #3 , #2) rectangle (#1 + #3, #2 + #4);
}

\newcommand{\stateTikzz}[7]{
\ifthenelse{#5 > 0 \AND #6>0}{\def\mycola{white}}{\def\mycola{black}};
\ifthenelse{#5 > 0}{\def\mycolb{black}}{\def\mycolb{white}};
\ifthenelse{#6 > 0}{\def\mycolc{black}}{\def\mycolc{white}};
\filldraw[color=\mycola, fill=\mycolb, #7, thick](#1- #3 , #2) 
    rectangle (#1 , #2 + #4);
\filldraw[color=\mycola, fill=\mycolc, #7, thick](#1, #2)
    rectangle (#1 + #3 , #2 + #4);
}

\newcommand{\stateTikzzz}[6]{
\ifthenelse{#5 = 0}{\def\mycola{black}\def\mycolb{white}\def\mycolc{white}}{};
\ifthenelse{#5 = 1}{\def\mycola{black}\def\mycolb{white}\def\mycolc{black}}{};
\ifthenelse{#5 = 2}{\def\mycola{black}\def\mycolb{black}\def\mycolc{white}}{};
\ifthenelse{#5 = 3}{\def\mycola{white}\def\mycolb{black}\def\mycolc{black}}{};
\ifthenelse{#5 = 4}{\def\mycola{grayCol}\def\mycolb{white}\def\mycolc{grayCol}}{};
\ifthenelse{#5 = 5}{\def\mycola{grayCol}\def\mycolb{grayCol}\def\mycolc{white}}{};
\ifthenelse{#5 = 6}{\def\mycola{white}\def\mycolb{black}\def\mycolc{grayCol}}{};
\ifthenelse{#5 = 7}{\def\mycola{white}\def\mycolb{grayCol}\def\mycolc{black}}{};
\ifthenelse{#5 = 8}{\def\mycola{white}\def\mycolb{grayCol}\def\mycolc{grayCol}}{};
\filldraw[color=\mycola, fill=\mycolb, #6, thick](#1- #3 , #2) 
    rectangle (#1 , #2 + #4);
\filldraw[color=\mycola, fill=\mycolc, #6, thick](#1, #2)
    rectangle (#1 + #3 , #2 + #4);
}

\newcommand{\linefive}[6]{
	\begin{scope}[shift={(0,#1)}]
        \stateTikzzz{0}{0}{1}{1}{#2}{dotted};
        \stateTikzzz{4}{0}{1}{1}{#3}{};
        \stateTikzzz{8}{0}{1}{1}{#4}{};
        \stateTikzzz{12}{0}{1}{1}{#5}{};
        \stateTikzzz{16}{0}{1}{1}{#6}{dotted};
    \end{scope}
}
\newcommand{\gaugefour}[5]{
    \begin{scope}[shift={(0,#1)}]
        \color{gaugeCol}
        \gaugestate{2}{0.1}{0.4}{0.8}{#2};
        \gaugestate{6}{0.1}{0.4}{0.8}{#3};
        \gaugestate{10}{0.1}{0.4}{0.8}{#4};
        \gaugestate{14}{0.1}{0.4}{0.8}{#5};
    \end{scope}
}

\newcommand{\evenline}[2]{
	\begin{scope}[shift={(0,#1)}]
        \stateTikz{0}{0}{1}{1}{#2}{dotted};
        \stateTikz{8}{0}{1}{1}{#2}{};
        \stateTikz{16}{0}{1}{1}{#2}{dotted};
    \end{scope}
}
\newcommand{\evenout}[2]{
	\begin{scope}[shift={(0,#1)}]
        \draw[color=#2] (1,1) -- (3,3)
                        (9,1) -- (11,3)
                        (7,1) -- (5,3)
                        (15,1) -- (13,3);
    \end{scope}
}
\newcommand{\evenouttwo}[2]{
    \begin{scope}[shift={(0,#1)}]
        \draw[color=#2] (1,1) -- (3.9,2) -- (3.9,3)
                        (7,1) -- (4.1,2) -- (4.1,3)
                        (9,1) -- (11.9,2) -- (11.9,3)
                        (15,1) -- (12.1,2) -- (12.1,3);
        \filldraw[color=#2, fill=white] (4,2) circle (0.3);
        \filldraw[color=#2, fill=white] (12,2) circle (0.3);
    \end{scope}
}

\newcommand{\oddline}[2]{
	\begin{scope}[shift={(0,#1)}]
        \stateTikz{4}{0}{1}{1}{#2}{};
        \stateTikz{12}{0}{1}{1}{#2}{};
    \end{scope}
}
\newcommand{\oddout}[2]{
	\begin{scope}[shift={(0,#1)}]
        \draw[color=#2] (1,3) -- (3,1)
                        (9,3) -- (11,1)
                        (7,3) -- (5,1)
                        (15,3) -- (13,1);
    \end{scope}
}

\newcommand{\oddouttwo}[2]{
    \begin{scope}[shift={(0,#1)}]
        \draw[color=#2, dashed] (-.1, 3) -- (-.1,2) -- (-1,1.6)
                                (16.1,3) -- (16.1,2) -- (17,1.6);
        \draw[color=#2] (3,1) -- (0.1,2) -- (0.1,3)
                        (5,1) -- (7.9,2) -- (7.9,3)
                        (11,1) -- (8.1,2) -- (8.1,3)
                        (13,1) -- (15.9,2) -- (15.9,3);

        \filldraw[color=#2, fill=white] (0,2) circle (0.3);
        \filldraw[color=#2, fill=white] (8,2) circle (0.3);
        \filldraw[color=#2, fill=white] (16,2) circle (0.3);
    \end{scope}
}

\newcommand{\gauge}[1]{
    \begin{scope}[shift={(0,#1)}]
        \foreach \i in {2,6,...,14}{
            \stateTikz{\i}{0.1}{0.4}{0.8}{gaugeCol}{};
        }
    \end{scope}
}
\newcommand{\gaugeout}[1]{
    \begin{scope}[shift={(0,#1)}]
        \foreach \i in {2,6,...,14}{
            \draw[color=gaugeCol] (\i,0.9) -- (\i, 3.1);
        }
    \end{scope}
}

\newcommand{\evolution}[1]{
    \begin{scope}[shift={(0,#1)}]
        \foreach \i in {2,6,...,14}{
            \filldraw[color=evolutionCol, fill=white] (\i,2) circle (0.3);
        }
    \end{scope}
}

\newcommand{\grid}[2]{
    \begin{scope}[shift={(#1,#2)}]
        \draw[->, very thick] (-1.5,-0.5) -- coordinate (x axis mid) (17.5,-0.5);
        \draw[->, very thick] (-1.5,-0.5) -- coordinate (y axis mid) (-1.5,7.5);
        \foreach \x in {0,4,...,16}
            \draw[very thick] (\x,-0.4) -- (\x,-0.7);
        \foreach \y in {0.5,3.5,...,6.5}
            \draw[very thick] (-1.7,\y) -- (-1.4,\y); 
        \node at (18,-0.5) {\Huge $x$};
        \node at (-1.5,8) {\Huge $t$};
    \end{scope}
}

\newcommand{\gridgauge}[2]{
    \begin{scope}[shift={(#1,#2)}]
        \draw[->, very thick] (-1.5,-0.5) -- coordinate (x axis mid) (17.5,-0.5);
        \draw[->, very thick] (-1.5,-0.5) -- coordinate (y axis mid) (-1.5,7.5);
        \foreach \x in {0,4,...,16}
            \draw[very thick] (\x,-0.4) -- (\x,-0.7);
        \foreach \x in {2,6,...,14}
            \draw[very thick] (\x,-0.5) -- (\x,-0.6);
        \foreach \y in {0.5,3.5,...,6.5}
            \draw[very thick] (-1.7,\y) -- (-1.4,\y); 
        \node at (18,-0.5) {\Huge $x$};
        \node at (-1.5,8) {\Huge $y$};
    \end{scope}
}

\definecolor{gaugeCol}{rgb}{0.40390625,0.6109375,0.09109375}
\definecolor{evolutionCol}{rgb}{0.3,0.3,0.3}

\begin{tikzpicture}

\evenline{0}{black};
\oddline{3}{black};
\evenline{6}{black};

\oddline{0}{black}
\evenline{3}{black}
\oddline{6}{black}

\draw[color=red,thick] (8,3.5) circle (4);


\draw[color=gaugeCol,dashed,thick](-0.60,3.2) rectangle (0,3.8);\draw[color=gaugeCol,dashed,thick](0,3.2) rectangle (0.60,3.8);

\draw[color=gaugeCol,thick](3.4,3.2) rectangle (4,3.8);\draw[color=gaugeCol,thick](4,3.2) rectangle (4.6,3.8);

\draw[color=gaugeCol,thick](7.4,3.2) rectangle (8,3.8);\draw[color=gaugeCol,thick](8,3.2) rectangle (8.6,3.8);

\draw[color=gaugeCol,thick](11.4,3.2) rectangle (12,3.8);\draw[color=gaugeCol,thick](12,3.2) rectangle (12.6,3.8);

\draw[color=gaugeCol,dashed,thick](15.4,3.2) rectangle (16,3.8);\draw[color=gaugeCol,dashed,thick](16,3.2) rectangle (16.6,3.8);


\draw[color=gaugeCol,dashed,thick](-0.60,6.2) rectangle (0,6.8);\draw[color=gaugeCol,dashed,thick](0,6.2) rectangle (0.60,6.8);

\draw[color=gaugeCol,thick](3.4,6.2) rectangle (4,6.8);\draw[color=gaugeCol,thick](4,6.2) rectangle (4.6,6.8);

\draw[color=gaugeCol,thick](7.4,6.2) rectangle (8,6.8);\draw[color=gaugeCol,thick](8,6.2) rectangle (8.6,6.8);

\draw[color=gaugeCol,thick](11.4,6.2) rectangle (12,6.8);\draw[color=gaugeCol,thick](12,6.2) rectangle (12.6,6.8);

\draw[color=gaugeCol,dashed,thick](15.4,6.2) rectangle (16,6.8);\draw[color=gaugeCol,dashed,thick](16,6.2) rectangle (16.6,6.8);


\draw[color=gaugeCol,dashed,thick](-0.60,0.2) rectangle (0,0.8);\draw[color=gaugeCol,dashed,thick](0,0.2) rectangle (0.60,0.8);

\draw[color=gaugeCol,thick](3.4,0.2) rectangle (4,0.8);\draw[color=gaugeCol,thick](4,0.2) rectangle (4.6,0.8);

\draw[color=gaugeCol,thick](7.4,0.2) rectangle (8,0.8);\draw[color=gaugeCol,thick](8,0.2) rectangle (8.6,0.8);

\draw[color=gaugeCol,thick](11.4,0.2) rectangle (12,0.8);\draw[color=gaugeCol,thick](12,0.2) rectangle (12.6,0.8);

\draw[color=gaugeCol,dashed,thick](15.4,0.2) rectangle (16,0.8);\draw[color=gaugeCol,dashed,thick](16,0.2) rectangle (16.6,0.8);


\draw (7.7,3.5) node {W};
\draw (8.3,3.5) node {S};
\draw (8.3,6.5) node {N};
\draw (11.7,3.5) node {E};

\gridgauge{0}{0}

\end{tikzpicture}}
    \caption{\ldots encoded in the general gauge extension layout.}
    \end{subfigure}
    \caption{Capturing the link model used for relative gauge extensions with the general definition.}
    \label{Equivalent centered GE}
    \label{fig:Equivalence between two different type of GE}
\end{figure}



The following specialises the previous, mathematical notion of gauge extension, to the restricted way in which it is understood in Physics:
\begin{definition}[Relative gauge extension]\label{def : relative gauge extension}
    Given a CA $F$ and a local gauge-transformation group $\Gamma$ of radius $s=0$, we say that a gauge extension $(F',\Gamma')$ is {\em relative} when:
    \begin{alphaenumerate}
        \item the gauge field is positioned on the links $(x,x+e_d)$, where $e_d$ takes values in \\$\{(1\ 0\ 0 \ldots), (0\ 1\ 0\ \ldots),\ldots\}$.
        \item the gauge field takes values in $G$---i.e. $\Lambda=G$
        \item for every position $x$ the gauge-transformations $\Gamma'$ act both on the matter field at $x$ according to $g_x\in G$, and on the gauge fields of its links, as follows:
        \begin{equation}\label{eq:defgammaext}
            \begin{cases}
                g_{x}(a)_{(x-e_d,x)} &= g_{x} \circ a_{(x-e_d,x)}\\
                g_{x}(a)_{(x,x+e_d)} &= a_{(x,x+e_d)} \circ g_{x}^{-1}
            \end{cases}
        \end{equation}
    \end{alphaenumerate}
\end{definition}
Thus relative extension keeps track of the difference of gauge between two neighbouring cells.

The above definitions were given for the $\mathbb{Z}^d$ grid, in order to establish the notion of gauge extension in full generality. The next section, however, will be given just in one dimension ($d=1$) for clarity. We have established it in arbitrary dimension $d$ in a private manuscript.

\section{Globally symmetric CA admit a relative gauge extension}
\label{sec:extension1}

From a Physics perspective, the gauge symmetry one seeks to impose usually comes from an already existing global symmetry. We show here that there is an equivalence between being globally $G$-symmetric and having a gauge extension with respect to $G$ a subgroup of the permutations of $\Sigma$.

\begin{theorem}[Global symmetry and relative gauge extension]\label{theo:globaltolocal}
    Let $F$ be a CA over alphabet $\Sigma$, $G$ a subgroup of the permutations of $\Sigma$ and $\Gamma$ the set of gauge-transformations defined using $G$ as the group of local gauge-transformations. Then the following two properties are equivalent:
    \begin{romanenumerate}
        \item \label{theo:piblind1}$F$ is globally $G$-symmetric
        \item \label{theo:piblind2} $(F,\Gamma)$ admits a relative gauge extension $(F',\Gamma')$ with the identity for the gauge field evolution, such that $F'$ commutes with any element of $\Gamma'$ (stronger than gauge-invariance because it does not require a $Z$-map).
    \end{romanenumerate}
\end{theorem}

\begin{proof}
~\\ (\ref{theo:piblind1} $\Rightarrow$ \ref{theo:piblind2})
Let $f$ be the local rule of $F$ with radius $r$.
Let $F'$ be a CA of radius $r$ over the extended configurations---containing a gauge field in between neighboring cells---such that the gauge field evolution is the identity and 
the local rule $f'$ for the evolution of the matter field is defined as follows: 
\begin{align}\label{eq:gifrompiblind}
    f'(c_{-r}, &a_{(-r+1,-r)}, ..., c_0, a_{(0,1)}, ...,c_r) =\nonumber \\
    &f\left( \prod_{i=-1}^{-r} a_{(i, i+1)} (c_{-r}), ..., a_{(-1,0)}(c_{-1}),\  c_0,\      a_{(0,1)}^{-1}(c_1), ...,\prod_{i=1}^{r} a_{(i-1, i)}^{-1} (c_{r})\right)
\end{align}
where $a$ is the gauge field and 
\begin{align*}
    \prod_{i=-1}^{-j} a_{(i, i+1)}(c_{-j}) &= a_{(-1,0)} \circ \ldots \circ a_{(-j+1,-j+2)} \circ a_{(-j,-j+1)}(c_{-j})\\
    \prod_{i=1}^{j} a_{(i-1, i)}^{-1}(c_{j}) &= a_{(0,1)}^{-1} \circ \ldots \circ a_{(j-2,j-1)}^{-1} \circ a_{(j-1,j)}^{-1}(c_{j}).
\end{align*}
$f'$ is defined so as to apply $f$ on a configuration that is on the same gauge basis---i.e. where relative gauge-transformations are cancelled and every cell is looked through the eyes of the gauge at position $0$.

We shall now prove that $(F',\Gamma')$---with $\Gamma'$ defined through Eq.\eqref{eq:defgammaext}---is a relative gauge extension of $(F,\Gamma)$.

The fact that this extension is relative is immediate from the definition. We therefore need to prove that this extension has the 3 required properties from definition \ref{def : Absolute-gauge extension}.

\begin{itemize}
    \item \emph{(Simulation)} The fact that $F'$ simulates $F$ when the gauge field is the identity is immediate from the definition of $f'$.
    \item \emph{(Extension)} Because $\Gamma'$ is defined through definition \ref{def : relative gauge extension}, it is also immediate that it verifies the extension property.
    \item \emph{(Gauge-invariance)} For any $\gamma'\in\Gamma'$ we will check that $\gamma'\circ F' = F'\circ \gamma'$.
    
    For $j$ between $1$ and $r$, we apply the gauge-transformation on the inputs ($c$ and $a$)---using  Eq.\eqref{eq:defgammaext}---and obtain by simple computation the following: 
    \begin{equation*}
        \begin{cases}
            \prod_{i=1}^{j} \left[ \gamma_{i} \circ a_{(i-1, i)} \circ \gamma_{i-1}^{-1}\right]^{-1} \gamma_j(c_{j}) &=  \gamma_{0} \circ \left[\prod_{i=1}^{j} a_{(i-1, i)}^{-1}(c_{j})\right]\\
            \prod_{i=-1}^{-j} \left[ \gamma_{i+1} \circ a_{(i, i+1)} \circ \gamma_{i}^{-1}\right] \gamma_j(c_{j}) &=  \gamma_{0} \circ \left[\prod_{i=-1}^{-j} a_{(i, i+1)}(c_{j})\right]
        \end{cases}        
    \end{equation*} where $\gamma_i \in G$. Therefore using this and Eq.\eqref{eq:gifrompiblind} we have that 
    \begin{equation}
        \left(F'\circ \gamma' (c,a)\right)_0 =
        \left(f\circ \gamma_0 \left( \begin{matrix}
            \prod_{i=-1}^{-r} a_{(i, i+1)} (c_{-r})\\
            ...\\
            a_{(-1,0)}(c_{-1}) \\
            c_0\\
            a_{(0,1)}^{-1}(c_1)\\
            ... \\
            \prod_{i=r}^{1} a_{(i-1, i)}^{-1} (c_{r})
        \end{matrix}\right)\right)_0
    \end{equation} where $\gamma_0$ is here applied to every element of the tuple.
    Since $F$ is globally $G$-symmetric, we have that $f\circ \gamma_0 = \gamma_0 \circ f$ and  therefore $ \left(F'\circ \gamma' (c,a)\right)_0 = \gamma_0 \circ \left(F'(c,a)\right)_0$ which finishes the proof that $F'\circ \gamma' = \gamma'\circ F'$ through translation invariance and because the gauge field evolution is the identity.
\end{itemize}

~\\ (\ref{theo:piblind2} $\Rightarrow$ \ref{theo:piblind1}) Suppose that $(F',\Gamma')$ is a relative gauge extension of $(F,\Gamma)$, such that $F'$ commutes with any element of $\Gamma'$, we shall prove that $F$ is globally $G$-symmetric (with $\Gamma$ the gauge-transformation group based on $G$).
Let $c$ be a configuration and $e$ denote the empty configuration of the gauge field. For any local gauge-transformation $g$, we write $\bar{\gamma}$ the global gauge-transformation applying $g$ everywhere---$g$ denotes both the element of $G$ and $G'$ depending on the context:
\begin{align*}
    \bar{\gamma}\circ F'(c,e) &= \bar{\gamma}(F(c),a) \tag{Simulation \ref{def : Absolute-gauge extension}}\\
        &= (\bar{\gamma}(F(c)),a') \tag{Extension \ref{def : Absolute-gauge extension}}
\end{align*}
where $a$ and $a'$ can be any gauge field configuration depending on $F'$ and $\bar{\gamma}$. And 
\begin{align*}
    F'\circ\bar{\gamma}(c,e) &= F'(\bar{\gamma}(c),\bar{\gamma}\circ e\circ\bar{\gamma}^{-1}) \tag{Extension \ref{def : relative gauge extension}}\\
        &= F'(\bar{\gamma}(c),e) \\
        &= (F(\bar{\gamma}(c)),a')  \tag{Simulation \ref{def : Absolute-gauge extension}}
\end{align*}
where $a$ and $a'$ can be any gauge field configuration depending on $F'$ and $\gamma$.
The $G'$-gauge-invariance of $F'$ give us $\bar{\gamma}\circ F'(c,e) = F'\circ\bar{\gamma}(c,e)$ and thus 
$$\bar{\gamma}(F(c)) = F(\bar{\gamma}(c)).$$
Therefore $F$ is globally $G$-symmetric.
\end{proof}

This theorem proves useful when looking for relative gauge extensions: first search for a global symmetry. The construction will be used in Sec. \ref{sec:universal} to prove that relative gauge extensions of CA are universal.

\section{Non-globally symmetric CA still admit an absolute gauge extension}
\label{sec:extension2}

We now prove that any CA can be intrinsically simulated by a gauge-invariant one, with respect to any gauge-transformation group, of any radius. The construction of this section uses non-relative gauge extensions, but it allows us to get rid of the prior requirements that there be a global symmetry or that the gauge-transformations be of radius $0$. 

\begin{theorem}[Every CA admits a gauge extension]\label{th:everyCA}
For any CA $F$ and gauge-transformation group $\Gamma$ there exists for some gauge field alphabet a gauge extension $(F',\Gamma')$. Furthermore the local rule of $F'$ acts as the identity over the gauge field.
\end{theorem}

\begin{proof}
We give here a constructive proof for any CA over $\mathbb{Z}^d$.

Let $F$ be a CA of radius $s'$ and $G$ be a local gauge-transformation group of radius $s$. We denote $r$ the highest radius between $s$ and $s'$. In the following we will consider neighbourhoods $R_x^k=[x-k\cdot r,x+k\cdot r]^d$ of each point $x\in\mathbb{Z}^d$, with $[a,b]=\{n\in \mathbb{Z}\ |\ a\leq n\leq b\}$.

First we choose $G$ as gauge field alphabet and define the effect of a global gauge-transformation $\gamma_x$, which applies $g\in G$ around $x$ according to $\gamma_x(a)_x=g\circ a_x$,
where $a$ denotes a gauge field configuration. The definition is so that the gauge field simply keeps track of every gauge-transformation applied around $x$. For any other cell of the gauge field, $\gamma_x$ has no impact. This condition along with the extension property of Definition \ref{def : Absolute-gauge extension} fully defines the new gauge-transformation group $G'$.

Next we define a new local rule $f'$ over the neighbourhood $R^5_x$. The definition below just states that the local rule applies $\prod_{i\in R^2_x}a_i^{-1}$ to undo all previous gauge-transformations, it then computes the evolution of $f$, and finally reapplies all the gauge-transformations, i.e.  
$$f'\big(c_{|R_x^5},a_{|R_x^5}\big)=\Big(\prod_{i\in R_x^1}a_i\circ f_{|R_x^2}\big(\prod_{i\in R_x^4}a_i^{-1}(c_{|R_x^5})_{|R_x^3}\big)\ ,\  a_x\Big)$$
where $f_{|R_x^2}$ denotes the function from $R_x^3$ to $R_x^2$ which calculates the temporal evolution of our automaton.

We can rewrite this local rule globally, using the notation $a$ to denote either the gauge field or a gauge-transformation which applies $a_x$ around each position $x$:
$$F'(c,a)_x=\Big(a\circ F\circ a^{-1}(c),a\Big)_x$$

Let us check that $(F',\Gamma')$ is a gauge extension:
\begin{itemize}
    \item \emph{(Simulation)} When the gauge field is the identity $f'$ acts the same as $f$ over the matter field, and as the identity over the gauge field.
    \item \emph{(Extension)} We used this property to define $G'$.
    \item \emph{(Gauge-invariance)} For any $\gamma'\in\Gamma'$---where $\Gamma'$ is built from $G'$ through definition \ref{def : Gauge-transformations}---we must check that $\gamma'\circ F' = F'\circ \gamma'$. 
    We reason globally to simplify notations:
    
    \begin{align*}
        F'\circ\gamma'(c,a)&=F'(\gamma(c),\gamma(a)) \tag{Extension \ref{def : Absolute-gauge extension}}\\
        &= \Big(\gamma(a)\circ F\circ{\gamma(a)^{-1}}(\gamma(c)),\ \ \gamma(a)\Big)\tag{Definition of $F'$}\\
        &= \Big(\gamma\circ a\circ F\circ {a^{-1}}\circ{\gamma^{-1}}\circ\gamma \big(c\big),\ \ \gamma(a)\Big)\tag{Definition \ref{def : Local gauge-transformation group}}\\
        &= \Big(\gamma\circ a\circ F\circ {a^{-1}}\big(c\big),\ \ \gamma(a)\Big)\tag{Definition \ref{def : Local gauge-transformation group}}
    \end{align*}
    
    \begin{align*}
        \gamma'\circ F'(c,a)&=\gamma' \Big(a\circ F\circ {a^{-1}}(c),\ \ a\Big)\tag{Definition of $F'$}\\
        &= \Big(\gamma\circ a\circ F\circ {a^{-1}}\big(c\big),\ \ \gamma(a)\Big) \tag{Extension \ref{def : Absolute-gauge extension}}
    \end{align*}
\end{itemize}
\end{proof}

\section{(Relative) gauge-invariant CA are universal}
\label{sec:universal}

Results in this section are only given for dimension 1.

In \cite{salo2013color}, the authors prove that for any alphabet $\Sigma$ containing 2 symbols or more, there exists an intrinsically universal globally $G$-symmetric cellular automaton on $\Sigma^\mathbb{Z}$, where $G$ is the group of all permutations of $\sigma$. The proof involves an extension which encodes the information in the structure of the configuration rather than the states, the idea being that a global transformation will conserve the structure---thus the information. Combining this result and Th. \ref{theo:globaltolocal}, we can easily prove the following corollary:
\begin{corollary}[Relative gauge-invariant cellular automata are universal]
    For any alphabet $\Sigma$ with $|\Sigma|\geq 2$, any gauge transformation group $G$ and any cellular automaton $F$ on $\Sigma^\mathbb{Z}$, there exists a $G'$-gauge-invariant CA $F'$ which intrinsically simulates $F$, with $G'$ the extended gauge-transformation based on $G$. Moreover, $F'$ arises as the relative gauge extension of a CA.
\end{corollary}

\begin{proof}
    Let $F''$ be a globally $G$-symmetric CA on $\Sigma^\mathbb{Z}$ that intrinsically simulates $F$ using 
    \cite[Theorem~1]{salo2013color}. 
    From Th. \ref{theo:globaltolocal}, $(F'',G)$ admits a relative gauge extension $(F',G')$ with 
    the evolution of the gauge field being the identity. Thus
    $F'$ intrinsically simulates $F''$, from which it follows that
    $F'$ is a $G'$-gauge-invariant CA which intrinsically simulates $F$.
\end{proof}

Such result is interesting on two accounts: (i) it shows that universality only requires relative gauge information and does not need any absolute information stored in the gauge field; (ii) it shows that relative gauge extensions, which are the ones usually appearing in Physics, are universal. Still, the universality of gauge-invariant CA is an even more direct corollary of Th. \ref{th:everyCA}. With that construction we can just pick any universal CA $F$, any local transformation group $G$ of any radius, and gauge-extend $F$ into $F'$. $F'$ acts trivially on the gauge field in this construction, it thus intrinsically simulates $F$ and is therefore universal.

\section{Sourcing the gauge field with the matter field}\label{sec:sourcing}

In both the construction of Th. \ref{theo:globaltolocal} and Th. \ref{th:everyCA}, the evolution rule of the gauge field to the identity, meaning that it does not evolve with time. It is often the case in Physics that a further twist is then introduced, so that the the matter field now influences the gauge field. We wish to do the same and find a gauge-extended CA whose gauge field influences the matter field, and whose matter field backfires on the gauge field. 

We use here the general definition of a gauge extension (Definition \ref{def : Absolute-gauge extension}) to search a gauge extension $F'$ of a non gauge-invariant CA $F$. Without loss of generality, $F' = (F'_1,F'_2)$, where $F'_1$ takes $(c,a)$ as input and returns the matter field after one time-step, and $F'_2$ does the same for the gauge field. We impose that the gauge (resp. matter) field be sourced by the matter (resp. gauge) field, in the strongest possible manner, i.e. we ask for $F'_2$ (resp. $F'_1$) to be injective in its first (resp. second) parameter.

We begin by choosing the alphabet $\Sigma=\{0,1,2\}^2$ and the space $\mathbb{Z}$ and we denote by $c_i^l$ and $c_i^r$ respectively the left and the right part of the cell.
In the following definitions we consider that all the additions and all the subtractions are modulo 3.

\noindent We define the initial automaton by the local rule:
$F(c)_i=(c_{i-1}^l-c_{i}^r,c_i^l+c_{i+1}^r)$, cf. Fig.\ref{fig:localrule of f}.

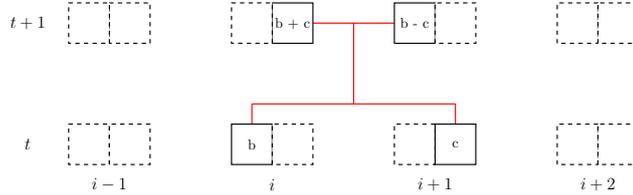
\begin{figure}[ht!]
    \centering
    \resizebox{0.6\textwidth}{!}{\begin{tikzpicture}
    
    \draw[color=red,thick] (-0.5,1.5) -- (4.5,1.5)
        (-0.5,1.5) -- (-0.5,1)
        (4.5,1.5) -- (4.5,1)
        (2,1.5) -- (2,3.5);
    \draw[color=red,thick] (1,3.5) -- (3,3.5);

    
    \filldraw[color=black,dashed, fill=white, thick](-4,0) rectangle (-5,1);
    \filldraw[color=black,dashed, fill=white, thick](-4,0) rectangle (-3,1);
    
    \filldraw[color=black, fill=white, thick](-1,0) rectangle (0,1);
    \draw (-0.5,0.5) node {b};
    \filldraw[color=black,dashed, fill=white, thick](0,0) rectangle (1,1);

    \filldraw[color=black,dashed, fill=white, thick](3,0) rectangle (4,1);
    \filldraw[color=black, fill=white, thick](4,0) rectangle (5,1);
    \draw (4.5,0.5) node {c};
    
    \filldraw[color=black,dashed, fill=white, thick](7,0) rectangle (8,1);
    \filldraw[color=black,dashed, fill=white, thick](8,0) rectangle (9,1);
    
    
    \filldraw[color=black,dashed, fill=white, thick](-3,3) rectangle (-4,4);
    \filldraw[color=black,dashed, fill=white, thick](-4,3) rectangle (-5,4);
    
    \filldraw[color=black, fill=white, thick](0,3) rectangle (1,4);
    \draw (0.5,3.5) node {b + c};
    \filldraw[color=black,dashed, fill=white, thick](0,3) rectangle (-1,4);

    \filldraw[color=black, fill=white, thick](3,3) rectangle (4,4);
    \draw (3.5,3.5) node {b - c};
    \filldraw[color=black,dashed, fill=white, thick](4,3) rectangle (5,4);
    
    \filldraw[color=black,dashed, fill=white, thick](7,3) rectangle (8,4);
    \filldraw[color=black,dashed, fill=white, thick](8,3) rectangle (9,4);
  
    \draw (4, -0.5) node {\Large $i+1$};
    \draw (0,-0.5) node {\Large $i$};
    \draw (-4,-0.5) node {\Large $i-1$};
    \draw (8,-0.5) node {\Large $i+2$};
    
    \draw (-6, 0.5) node {\Large $t$};
    \draw (-6, 3.5) node {\Large $t+1$};

\end{tikzpicture}}
    \caption{The local rule of F}
    \label{fig:localrule of f}
\end{figure}

We consider a local group of gauge-transformation containing three elements, namely:
    $$ G=\{\sigma_0,\sigma_1,\sigma_2\}$$
    where $\sigma_i$ is the function of radius $0$ that adds $i$ to each part of the cell.
    
We can check that $F$ is not gauge-invariant for $\Gamma$ (as defined from $G$), by considering a configuration $c$ which associates $(1,1)$ to position $i$ and $(0,0)$ to all other positions. Let $\gamma$ be a gauge-transformation which applies $\sigma_2$ over $i$, $\gamma(c)$ is then the fully empty configuration $e$. Since $F$ preserves emptiness we have:
$$F\circ\gamma(c)=\gamma(c)=e$$
But when we apply $F$ to $c$ we obtain non-empty cells in $i+1$ and $i-1$, this contradicts the gauge-invariance definition. This idea is illustrated in sub-Figs \ref{subfigs: F} and \ref{subfigs: F gamma}.

\begin{figure}[ht!]
    \centering
    \begin{subfigure}[b]{0.49\textwidth}
    \centering
    \resizebox{\textwidth}{!}{ \includegraphics{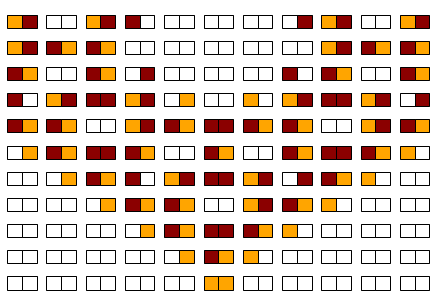}}
    \caption{$F$ over $c$}
    \label{subfigs: F}
    \end{subfigure}
    \hfill
    \begin{subfigure}[b]{0.49\textwidth}
    \centering
    \resizebox{\textwidth}{!}{ \includegraphics{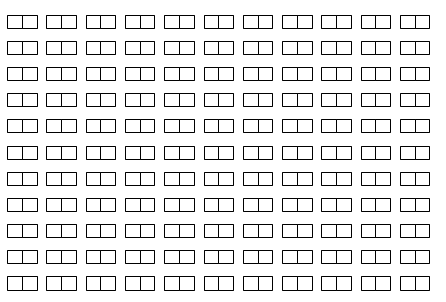}}
    \caption{$F$ over $\gamma(c)$}
    \label{subfigs: F gamma}
    \end{subfigure}
    
    \begin{subfigure}[b]{0.49\textwidth}
    \centering
    \resizebox{\textwidth}{!}{ \includegraphics{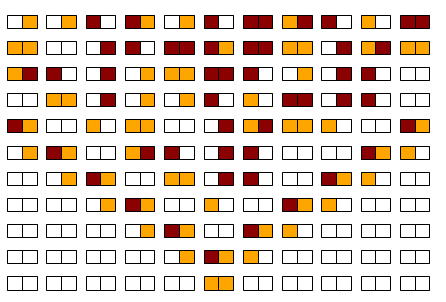}}
    \caption{$F'$ over $(c,e)$}
    \label{subfigs: F' }
    \end{subfigure}
    \hfill
    \begin{subfigure}[b]{0.49\textwidth}
    \centering
    \resizebox{\textwidth}{!}{ \includegraphics{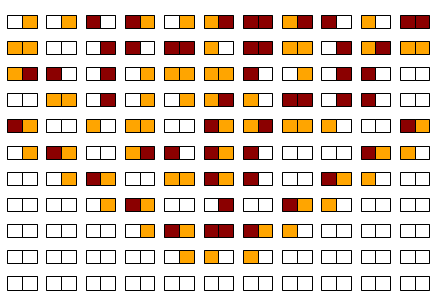}}
    \caption{$F'$ over $\gamma(c,e)$}
    \label{subfigs: F' gamma}
    \end{subfigure}
    \caption{Space-time representation of $F$ and $F'$ over the same initial configurations $c$ and $\gamma(c)$, where $c$ is the configuration at the bottom line of \ref{subfigs: F} \ref{subfigs: F' }. The values $1$ and $2$ are respectively represented by orange and red, while $0$ is just an empty cell. Only the matter field is represented here.}
    \label{fig:illustration GI}
\end{figure}

    We now provide a gauge extension for $(F,\Gamma)$. We begin by choosing the gauge field alphabet $\Lambda=\Sigma$ and placing the gauge field between each cell.

    Next we engineer the injective influence of the gauge field over the matter field in the simplest possible way. We simply add, to each sub-cell of the matter field, the value of the nearest sub-cell of the gauge field during each evolution. 
    See Fig.  \ref{figs: new local rule}.
    
    Finally we extend gauge-transformations to the gauge field (Fig.  \ref{figs: new gauge-transformation group}) and choose the evolution of the gauge field (green cells of Fig.  \ref{figs: new local rule}) to make sure that $F'$ is a reversible gauge-invariant dynamics. Notice that the figures where chosen so that all cases are given. The proof of gauge-invariance for this example is given in appendix-\ref{appendix: gauge-invariance}, and can be visually seen from Figs.  \ref{subfigs: F' } and \ref{subfigs: F' gamma} where a gauge-transformation does not impact the overall dynamics.

\begin{figure}[ht!]
    \centering
    \begin{subfigure}[b]{0.52\textwidth}
    \centering
    \resizebox{\textwidth}{!}{\definecolor{gaugeCol}{rgb}{0.40390625,0.6109375,0.09109375}
\begin{tikzpicture}

    \draw[thick] (-3,-1.5) -- (6,-1.5)
        (-3,3.5) -- (6,3.5);;
    \draw[thick,dashed] (-5,-1.5) -- (9,-1.5)
        (-5,3.5) -- (9,3.5);;
    
    \filldraw[color=gaugeCol, fill=white,ultra thick] (2, 3.5) circle (2) ;
    \draw[color=gaugeCol,ultra thick] (2,1.5) -- (2,5.5);
    \draw (1.15,3.5) node[scale=1.5] { $\color{gaugeCol}a_1\color{red}-k_2$};
    \draw (3,3.5) node[scale=1.5] {$\begin{matrix}
    \color{gaugeCol}a_2\color{red}-k_2\\
    \color{red}-k_1
    \end{matrix}$};

    \filldraw[color=gaugeCol, fill=white,ultra thick] (2, -1.5) circle (2) ;
    \draw[color=gaugeCol,ultra thick] (2,-3.5) -- (2,0.5);
    \draw[gaugeCol] (1.15,-1.5) node[scale=2] {$a_1$};
    \draw[gaugeCol] (3,-1.5) node[scale=2] {$a_2$};

    \filldraw[color=black, fill=white, ultra thick](-2.25,-2.25) rectangle (-3.75,-0.75);
    \draw (-3,-1.5) node[scale=2] {$c_1^l$};
    \filldraw[color=black,  fill=white, ultra thick](-0.75,-2.25) rectangle (-2.25,-0.75);
    \draw (-1.5,-1.5) node[scale=2] {$c_1^r$};

    \filldraw[color=black, fill=white, ultra thick](4.75,-2.25) rectangle (6.25,-0.75);
    \draw (5.5,-1.5) node[scale=2] {$c_2^l$};;
    \filldraw[color=black, fill=white,ultra thick](6.25,-2.25) rectangle (7.75,-0.75);
    \draw (7,-1.5) node[scale=2] {$c_2^r$};

    \filldraw[color=black, fill=white, ultra thick](-2.25,2.75) rectangle (-3.75,4.25);
    \draw (-3,3.5) node[scale=1.2] {$c_1^l\color{red}+k_1$};
    \filldraw[color=black,  fill=white, ultra thick](-0.75,2.75) rectangle (-2.25,4.25);
    \draw (-1.5,3.5) node[scale=1.2] {$c_1^r\color{red}+k_1$};

    \filldraw[color=black, fill=white, ultra thick](4.75,2.75) rectangle (6.25,4.25);
    \draw (5.5,3.5) node[scale=1.2] {$c_2^l\color{red}+k_2$};;
    \filldraw[color=black, fill=white, ultra thick](6.25,2.75) rectangle (7.75,4.25);
    \draw (7,3.5) node[scale=1.2] {$c_2^r\color{red}+k_2$};

    \draw[color=red] (-2.25, 1) ellipse (5 and 0.4);
    \draw[color=red] (6.25, 1) ellipse (5 and 0.4);
    \draw[color=red] (-2.25, 1) node[scale=2] {\Large $\sigma_{k_1}$};
    \draw[color=red] (6.25, 1) node[scale=2] {\Large $\sigma_{k_2}$};
    
    \draw[gray] (-6.5, -1.5) node[scale=2] {\Large $c$};
    \draw[gray] (-6.5, 3.5) node[scale=2] {\Large $\widetilde\gamma(c)$};

\end{tikzpicture}}
    \caption{The new gauge-transformation group $\Gamma'$}
    \label{figs: new gauge-transformation group}
    \end{subfigure}
    \hfill
    \begin{subfigure}[b]{0.47\textwidth}
    \centering
    \resizebox{\textwidth}{!}{ \definecolor{gaugeCol}{rgb}{0.40390625,0.6109375,0.09109375}
\begin{tikzpicture}

    \draw[thick] (-3,-1.5) -- (6,-1.5)
        (-3,3.5) -- (6,3.5);;
    \draw[thick,dashed] (-5,-1.5) -- (9,-1.5)
        (-5,3.5) -- (9,3.5);;
    
    \filldraw[color=gaugeCol, fill=white,ultra thick] (2, 3.5) circle (2) ;
    \draw[color=gaugeCol,ultra thick] (2,1.5) -- (2,5.5);
    \draw (1.15,3.2) node[scale=1.8] { \color{gaugeCol}$\begin{matrix}
            a_2 \\
            \color{black}+\\
            \color{black}c_{1}^l
        \end{matrix}$};
    \draw (3,3.2) node[scale=1.8] { \color{gaugeCol}$\begin{matrix}
            a_{1}+a_{2} \\
            \color{black}+\\
            \color{black}c_{2}^r
        \end{matrix}$};

    \filldraw[color=gaugeCol, fill=white,ultra thick] (2, -1.5) circle (2) ;
    \draw[color=gaugeCol,ultra thick] (2,-3.5) -- (2,0.5);
    \draw[gaugeCol] (1.15,-1.5) node[scale=2] {$a_1$};
    \draw[gaugeCol] (3,-1.5) node[scale=2] {$a_2$};

    \filldraw[color=black, fill=white, ultra thick](-2.75,-2.5) rectangle (-4.75,-0.5);
    \draw (-3.75,-1.5) node[scale=2] {$c_1^l$};
    \filldraw[color=black, dashed, fill=white, ultra thick](-0.75,-2.5) rectangle (-2.75,-0.5);
    \draw (-1.75,-1.5) node {\color{gray}$\ldots$};

    \filldraw[color=black,dashed, fill=white, ultra thick](4.75,-2.5) rectangle (6.75,-0.5);
    \draw (5.75,-1.5) node {\color{gray}$\ldots$};;
    \filldraw[color=black, fill=white,ultra thick](6.75,-2.5) rectangle (8.75,-0.5);
    \draw (7.75,-1.5) node[scale=2] {$c_2^r$};

    \filldraw[color=black,dashed, fill=white, ultra thick](-2.75,2.5) rectangle (-4.75,4.5);
    \draw (-3.75,3.5) node {\color{gray}$\ldots$};
    \filldraw[color=black,  fill=white, ultra thick](-0.75,2.5) rectangle (-2.75,4.5);
    \draw (-1.75,3.5) node[scale=1.5] {$\begin{matrix}
    c_{1}^l+c_{2}^r\\
    \color{gaugeCol}+\\
    \color{gaugeCol}a_{1}
    \end{matrix}$};

    \filldraw[color=black, fill=white, ultra thick](4.75,2.5) rectangle (6.75,4.5);
    \draw (5.75,3.5) node[scale=1.5] {$\begin{matrix}
    c_1^l-c_{2}^r\\
    \color{gaugeCol}+\\
    \color{gaugeCol}a_{2}
    \end{matrix}$};;
    \filldraw[color=black,dashed, fill=white, ultra thick](6.75,2.5) rectangle (8.75,4.5);
    \draw (7.75,3.5) node {\color{gray}$\ldots$};

    \draw[gray] (-6.5, -1.5) node[scale=2] {\Large $c$};
    \draw[gray] (-6.5, 3.5) node[scale=2] {\Large $F'(c)$};

\end{tikzpicture}}
    \caption{The new local rule $f'$}
    \label{figs: new local rule}
    \end{subfigure}
    \caption{Description of the gauge extension $(F',\Gamma')$. Green circles represent the gauge field and black rectangles the matter field.}
    \label{fig: Description of the gauge extension (F',Gamma')}
\end{figure}
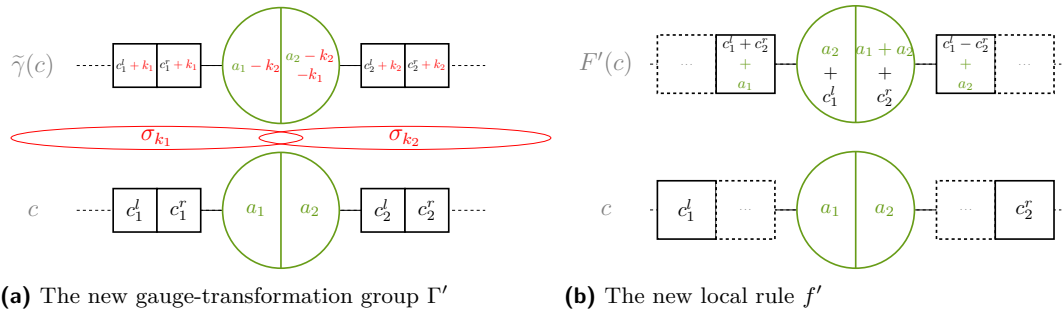


Overall, starting from a CA $F$ we have defined a gauge extension $F'$ which features a strong interaction between the gauge and the matter field. In the world of CA this is the first example of the kind \cite{ArrighiGauge,arrighi2019non}. Building this example required the choice of a very specific extension of the gauge-transformation over the gauge field (cf Fig.\ref{figs: new gauge-transformation group}) so as to obtain gauge-invariance whilst preserving reversibility and injectivity. Under a relative gauge extension this extension of the gauge-transformation is forced upon us, it seems hard to find such an example.

Notice that since the gauge field is sourced by the matter field it typically does not remain empty during the evolution. Thus $F'$ can only simulate $F$ for one time step. This may seem strange from a mathematical point of view, as we may expect from an extension that it preserves the original dynamics over several steps, too. But in Physics the initial non gauge-invariant theory is indeed used to inspire a more complex dynamics, which enriches and ultimately diverges from the original one. Fig.\ref{fig:9000 temporal step} shows how starting from the same configuration, one obtains very different evolutions.




\begin{figure}[ht!]
    \centering
    \begin{subfigure}[b]{0.49\textwidth}
    \centering
    \resizebox{\textwidth}{!}{ \includegraphics{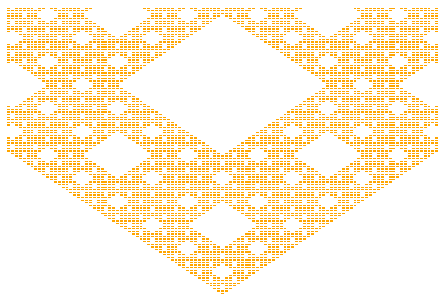}}
    \caption{Evolution of $F$}
    \label{GT_over_GF_a}
    \end{subfigure}
    \hfill
    \begin{subfigure}[b]{0.49\textwidth}
    \centering
    \resizebox{\textwidth}{!}{ \includegraphics{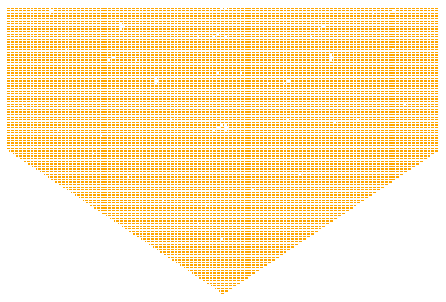}}
    \caption{Evolution of $F'$}
    \end{subfigure}
    \caption{Evolution of $F$ and its gauge extension $F'$ over 9000 temporal step}
    \label{GT_over_GF_b}
    \label{fig:9000 temporal step}
\end{figure}

\label{sec:interactions}

\section{Conclusion}

In order to obtain a gauge-invariant theory, starting from a non-gauge-invariant one, the usual route is to extend the theory by means of a gauge field. As discussed in the introduction, the gauge field usually turns out to be a connection between gauge choices at neighbouring points, but there is no immediate reason why this should be the case. In the first part, we formalised, in the framework of Cellular Automata (CA), the notions of gauge extension and relative gauge extension. The latter forces the gauge field to act as a connection. 

Again in Physics one usually starts from a theory featuring a global symmetry, before `making it local' through the gauge extension. Again there is no immediate reason why this should be the case. 
In our framework, we were able to establish a logical relation between global symmetry and relative gauge-invariance. Namely we proved that the CA that admit a relative gauge extension are exactly those that have the corresponding global symmetry. To the best of our knowledge, no continuous equivalent of that theorem exists in the literature; perhaps the discrete offers better opportunities for formalisation. 

We also proved that any CA can be extended into a gauge-invariant one. Thus, gauge-invariant CA are universal. Two different constructions were provided. The first construction uses the gauge field to store, at each location, the value of the gauge-transformation which the matter field has undergone at that location, thereby allowing for the action of the transformation to be counteracted. This path uses a non-relative gauge extension. The second construction puts together the fact that any CA can be made globally-symmetric \cite{salo2013color}, with the fact that any globally-symmetric CA admits a relative gauge extension. Thus, relative gauge-extended CA are universal.

Whilst the introduction of the gauge field is initially motivated by the gauge symmetry requirement, the gauge field ends up triggering new, richer behaviours as it influences the matter field. However, in order for it to mediate the interactions within the matter field, as is the case in Physics, it should be the case that the matter field also influences the gauge field---and back. In this paper, we provided a first example of a gauge-extended CA whose matter field injectively influences gauge field, whilst preserving reversibility. This was done through a general gauge extension, we leave it open whether this can be achieved through a relative gauge extension. The difficulty here is that relative gauge extensions seem to store just the minimal amount of information required for gauge-invariance, and any further influence upon the gauge field runs the risk of jeopardising that. 

This difficulty can be circumvented in the quantum setting: the Quantum Cellular Automaton of \cite{ArrighiQED} arises from a relative gauge extension, and yet features and a gauge field which is `sourced' by the matter field. The construction directly yields a quantum simulation algorithm for one-dimensional quantum electrodynamics. This should serve us a reminder that whilst this work is theoretical, it is not merely of theoretical interest. Gauge extensions is exactly what one needs to do in order to capture physical interactions within discrete quantum models. This may lead for instance to digital quantum simulation algorithms, with improved numerical accuracy, as fundamental symmetries are preserved throughout the computation. 



\bibliography{biblio}

\appendix

\section{Proof of gauge-invariance for Sec. \ref{sec:sourcing}}\label{appendix: gauge-invariance}

In order to prove that the example illustrated in Fig-\ref{fig: Description of the gauge extension (F',Gamma')} is gauge-invariant, we will show that $\gamma' \circ F' = F' \circ \gamma'$ for any $\gamma' \in \Gamma'$. It is sufficient to prove this locally, we do so using the notations of the figure and we denote by $f'$ and $g'$ the local application of the evolution and a gauge-transformation: 
\begin{align*}
    f'\circ g' (c_1^l,a_1,a_2,c_2^r)
        &= \begin{pmatrix*}[l]c_1^l + c_2^r + a_1 + k_1, \\ 
            a_2 +c_1^l -k_2, \\
            a_1+a_2+c_2^r-k_1-k_2, \\
            c_1^l-c_2^r + a_2+k_2\end{pmatrix*} \\
        &=g'\circ f' (c_1^l,a_1,a_2,c_2^r) 
\end{align*}

Therefore $F'$ is $\Gamma'$-gauge-invariant.

\end{document}